\documentclass{ws-ijfcs}
\usepackage{enumerate}
\usepackage{url}
\usepackage[utf8]{inputenc}
\usepackage[T1]{fontenc}
\usepackage{amssymb}
\usepackage{amsmath}
\usepackage{amsfonts}
\usepackage{xifthen}
\usepackage{graphicx}
\usepackage{enumerate}
\usepackage{tikz}
\usetikzlibrary{petri}
\usepackage{multicol}

\usepackage{complexity}
\usepackage{mathtools}
\usepackage{comment}
\usepackage{mathrsfs}
\usepackage{longtable}
\usepackage{enumitem}
\usepackage{etoolbox}
\usepackage{float}
\usepackage{color}
\usepackage{wasysym}
\usepackage{hyperref}

\renewcommand{\tt}{\hspace*{3mm}}

\usepackage{mathtools}

\newcommand*{\DEBUG}{}%

\ifdefined\DEBUG
\newcommand{\fixme}[1]{{\textcolor{red}{\bf{\textsf{FIXME: #1}}}}}
\newcommand{\bug}[1]{{\textcolor{blue}{\bf{\textsf{BUG: #1}}}}}
\newcommand{\idea}[1]{{\textcolor{blue}{\bf{\textsf{IDEA: #1}}}}}

\newcommand{\TODO}[1]{{\textcolor{red}{\bf{\textsf{ 
TODO: #1
}}}}}
\else
\newcommand{\fixme}[1]{}
\newcommand{\bug}[1]{}
\newcommand{\TODO}[1]{}
\newcommand{\idea}[1]{}
\fi

\newclass{\COMSLIP}{COM\mbox{-}SLIP}
\newclass{\COMSLIPCUP}{COM\mbox{-}SLIP^{\cup}}

\newclass{\DCM}{DCM}
\newclass{\eDCM}{eDCM}
\newclass{\eNPDA}{eNPDA}
\newclass{\DPDA}{DPDA}
\newclass{\RDPDA}{RDPDA}
\newclass{\PDA}{PDA}
\newclass{\DCMNE}{DCM_{NE}}
\newclass{\TwoDCM}{2DCM}
\newclass{\NCM}{NCM}
\newclass{\eNCM}{eNCM}
\newclass{\eNQA}{eNQA}
\newclass{\eNSA}{eNSA}
\newclass{\eNPCM}{eNPCM}
\newclass{\eNQCM}{eNQCM}
\newclass{\eNSCM}{eNSCM}
\newclass{\DPCM}{DPCM}
\newclass{\NPCM}{NPCM}
\newclass{\NQCM}{NQCM}
\newclass{\NSCM}{NSCM}
\newclass{\NPDA}{NPDA}
\newclass{\TRE}{TRE}
\newclass{\NFA}{NFA}
\newclass{\DFA}{DFA}
\newclass{\NCA}{NCA}
\newclass{\DCA}{DCA}
\newclass{\DTM}{DTM}
\newclass{\NTM}{NTM}
\newclass{\DLOG}{DLOG}
\newclass{\CFG}{CFG}
\newclass{\ETOL}{ET0L}
\newclass{\EDTOL}{EDT0L}
\newclass{\CFP}{CFP}
\newclass{\ORDER}{O}
\newclass{\MATRIX}{M}
\newclass{\BD}{BD}
\newclass{\LB}{LB}
\newclass{\ALL}{ALL}
\newclass{\decLBD}{decLBD}
\newclass{\StLB}{StLB}
\newclass{\SBD}{SBD}
\newclass{\TCA}{TCA}
\newclass{\RNCSA}{RNCSA}
\newclass{\RDCSA}{RDCSA}

\newclass{\DCSA}{DCSA}
\newclass{\NCSA}{NCSA}
\newclass{\DCSACM}{DCSACM}
\newclass{\NCSACM}{NCSACM}
\newclass{\NTMCM}{NTMCM}

\newclass{\CFREEL}{CFL}

\newclass{\UFIN}{U\mbox{-}IND}
\newclass{\UFINONE}{\LL(IND_{{UFIN}_1})}
\newclass{\FIN}{\LL(IND_{FIN})}
\newclass{\ILIN}{L\mbox{-}IND}
\newclass{\IRLIN}{RL\mbox{-}IND}
\newclass{\ETOLfin}{\LL(ET0L_{FIN})}

\newsavebox{\spacebox}
\begin{lrbox}{\spacebox}
\verb*! !
\end{lrbox}
\newcommand{\LL}{{\cal L}}
\newcommand{\MM}{{\cal M}}
\newcommand{\FF}{{\cal F}}

\DeclareMathOperator{\comm}{comm}

\begin{document}

\markboth{Oscar H. Ibarra, Ian McQuillan}
{Semilinearity of Families of Languages}

%
\catchline{}{}{}{}{}
%

\title{Semilinearity of Families of Languages\thanks{\textcopyright 2022. This manuscript version is made available under the CC-BY 4.0 license \url{https://creativecommons.org/licenses/by/4.0/}. Published at {\it International Journal of Foundations of Computer Science}, 31 (8), 1179--1198 (2020) \url{https://doi.org/10.1142/S0129054120420095}}}

\author{Oscar H. Ibarra}

\address{
Department of Computer Science, University of California, Santa Barbara, CA 93106, USA\\ \email{ibarra@cs.ucsb.edu}}

\author{Ian McQuillan}

\address{
Department of Computer Science, University of Saskatchewan
	Saskatoon, SK S7N 5A9, Canada\\
	\email{mcquillan@cs.usask.ca}
}

\maketitle

\begin{history}
\received{(Day Month Year)}
\accepted{(Day Month Year)}
\comby{(xxxxxxxxxx)}
\end{history}

\begin{abstract}
Techniques are developed for creating new and general language families of only semilinear languages, and for showing families only contain
semilinear languages.
It is shown that for language families $\LL$ that are semilinear full trios,
the smallest full
AFL containing $\LL$ that is also closed under intersection with languages in  $\NCM$
(where $\NCM$ is the family of languages accepted by
$\NFA$s augmented with reversal-bounded counters), is also semilinear.
If these closure properties are effective, this also immediately implies decidability of membership,
emptiness, and infiniteness for these general families.
From the general techniques, new grammar systems are given that are extensions of 
well-known families of semilinear full trios, whereby it is implied that these extensions must
only describe semilinear languages. This also implies positive decidability properties for the new systems.
Some characterizations
of the new families are also given.
\end{abstract}

\keywords{semilinearity; closure properties; counter machines; pushdown automata; decidability.}

\section{Introduction}

One-way nondeterministic reversal-bounded multicounter machines ($\NCM$) operate like $\NFA$s with $\lambda$ transitions, where there are some number
of stores that each can contain some non-negative integer. The transition function can detect whether each counter is zero or non-zero, and optionally increment or decrement each counter; however, there is a bound on the number of changes each counter can make between non-decreasing and non-increasing. These machines have been extensively studied in the literature, for example in \cite{Ibarra1978}, where it was shown that $\NCM$s only accept semilinear languages (defined in Section \ref{sec:prelims}). As the semilinear property is effective for $\NCM$ (in that, the proof consists of an algorithm for constructing a finite representation of the semilinear sets), this implies that $\NCM$s have
decidable membership, emptiness, and infiniteness properties, as emptiness and infiniteness can be decided easily on semilinear sets (and membership follows from emptiness by effective closure under intersection with regular languages). $\NCM$ machines have been applied extensively in the literature, for example, to model checking and verification \cite{IBARRA2002165,IbarraBultanSu,IbarraSu,IbarraBultan}, often using the positive decidability properties of the family. 

More general machine models have been studied with an unrestricted pushdown automaton augmented by some number of reversal-bounded counters ($\NPCM$, \cite{Ibarra1978}). Despite the unrestricted pushdown, the languages accepted are all semilinear, implying they have the same decidable properties. This family too has been applied to several verification problems \cite{DangIbarraBultan,IbarraDangTCS}, including model checking recursive programs with numeric data types \cite{HagueLin2011}, synchronization- and reversal-bounded analysis of multithreaded programs \cite{HagueLinCAV}, for showing decidable properties of models of integer-manipulating programs with recursive parallelism \cite{HagueLinRP}, and for decidability of problems on commutativity \cite{eDCM}. In these papers, the positive decidability properties --- the result of the semilinearity --- plus the use of the main store (the pushdown), plus the counters, played a key role. Hence, (effective) semilinearity is a crucial property for families of languages.

The ability to augment a machine model with reversal-bounded counters and to only accept semilinear languages is not unique to pushdown automata; in \cite{Harju2002278}, it was found that many classes of machines $\MM$ accepting semilinear languages could be augmented with reversal-bounded counters, and the resulting family
$\MM_c$ would also only accept semilinear languages. 
This includes models such as Turing machines with a one-way read-only input tape and a finite-crossing\footnote{A worktape is finite-crossing if there is a  bound on the number of times the boundary of all neighboring cells on the worktape are crossed.} worktape.
However, a precise formulation of which classes of machines this pertains to was not given.

Here, a precise formulation of families of languages that can be ``augmented'' with counters will be examined in terms of closure properties rather than machine models. This allows for application to families described by machine models, or grammatical models. It is shown that for any full trio (a family closed under homomorphism, inverse homomorphism, and intersection with regular languages) of semilinear languages $\LL_0$, then the smallest full AFL $\LL$ (a full trio also closed under union, concatenation, and Kleene-*) containing 
$\LL_0$ that is closed under intersection with languages in $\NCM$, must only contain semilinear languages. Furthermore, if the closure properties and semilinearity are effective in $\LL_0$, this implies a decidable membership, emptiness, and infiniteness problem in $\LL$. Hence, this provides a new method for creating general families of languages with positive decidability properties.

Several specific models are created by adding counters. For example, indexed grammars are a well-studied general grammatical model like context-free grammars except where nonterminals keep stacks of ``indices''. Although this system can generate
non-semilinear languages, linear indexed grammars (indexed grammars with at most one nonterminal in the right hand side of
every production) generate only semilinear languages \cite{DP}. Here, we define {\em linear indexed grammars with counters}, akin to linear indexed grammars, where every sentential form contains the usual sentential form, plus $k$ counter values; each production operates as usual and can also optionally increase each counter by some amount; and a terminal word can be generated only if it can be produced with all counter values equal. It is shown that the family of languages generated must be semilinear since it is contained in the smallest full AFL containing the intersection of linear indexed languages and $\NCM$ languages. A characterization is also shown: linear indexed grammars with counters generate exactly those languages obtained by intersecting a linear indexed language with an $\NCM$ and then applying a homomorphism. Furthermore, it is shown that right linear indexed grammars (where terminals only appear to the left of nonterminals in productions) with counters coincide exactly with the machine model $\NPCM$. Therefore, linear indexed grammars with counters are a natural generalization of $\NPCM$ containing only semilinear languages. This model is generalized once again as follows: an indexed grammar is uncontrolled finite-index if, there is a value $k$ such that, for every derivation in the grammar, there are at most $k$ occurrences of nonterminals in every sentential form. It is known that every uncontrolled finite-index indexed grammar generates only semilinear languages \cite{LATA2017journal,ICALP2015}. It is shown here that uncontrolled finite-index indexed grammars with counters generate only semilinear languages, which is also a natural generalization of both linear indexed grammars with counters and $\NPCM$. This immediately shows decidability of membership, emptiness, and infiniteness for this family.

Lastly, the closure property theoretic method of adding counters is found to often be more helpful than the machine model method of \cite{Harju2002278} in terms of determining whether the resulting family is semilinear, as here a machine model $\MM$ is constructed such that the language family accepted by $\MM$ is a semilinear full trio, but adding counters to the model to create $\MM_c$ accepts non-semilinear languages. This implies from our earlier results, that $\MM_c$ can accept languages that cannot be obtained from any accepted by $\MM$ by allowing any number of intersections with $\NCM$s combined with any of the full AFL operations.

This paper therefore contains useful new techniques for creating new language families, and for showing existing language families only contain semilinear languages, which can then be used to immediately obtain decidable emptiness, membership, and infiniteness problems. Such families can perhaps also be applied to various areas, such as to verification, similarly to the use of $\NPCM$. A preliminary version of this paper appeared in \cite{CIAA2018}. This version includes all missing proofs omitted due to space constraints, and the new Proposition \ref{fullAFLClosed} which allows for some of the other proposition statements to be more general. Section \ref{closure} is also new.

\section{Preliminaries}
\label{sec:prelims}

In this section, preliminary background and notation is given. 

Let $\mathbb{N}_0$ be the set of non-negative integers, and let $\mathbb{N}_0^k$ be the set of all $k$-tuples of non-negative integers. A set $Q \subseteq \mathbb{N}_0^k$ is {\em linear} if there exists vectors $\vec{v_0}, \vec{v_1}, \ldots, \vec{v_l} \in \mathbb{N}_0^k$ such that $Q = \{ \vec{v_0} + i_1 \vec{v_1} + \cdots + i_l \vec{v_l} \mid i_1, \ldots, i_l \in \mathbb{N}_0\}$. Here, $\vec{v_0}$ is called the {\em constant}, and $\vec{v_1}, \ldots, \vec{v_l}$ are called the {\em periods}. A set $Q$ is called semilinear if it is a finite union of linear sets.

Introductory knowledge of formal language and automata theory is assumed such as nondeterministic finite automata ($\NFA$s), pushdown automata ($\NPDA$s), Turing machines, and closure properties
\cite{HU}. An {\em alphabet} $\Sigma$ is a finite set of symbols, a {\em word} $w$ over $\Sigma$ is a finite sequence of symbols from $\Sigma$, and $\Sigma^*$ is the set of all words over $\Sigma$ which includes the empty word $\lambda$. 
A {\em language} $L$ over $\Sigma$ is any $L \subseteq \Sigma^*$. The {\em complement} of a language $L \subseteq \Sigma^*$, denoted by $\overline{L}$, is $\Sigma^* - L$.

Given a word $w\in \Sigma^*$, the length of $w$ is denoted by $|w|$. For $a\in \Sigma$, the number of $a$'s in $w$ is denoted by $|w|_a$. 
 Given a word $w$ over an alphabet $\Sigma = \{a_1, \ldots, a_k\}$, the Parikh image of $w$ is $\psi(w) = (|w|_{a_1}, \ldots, |w|_{a_k})$, and the Parikh image of a language $L$ is $\{\psi(w) \mid w \in L\}$. The commutative closure of a language $L$ is the language $\comm(L) = \{w \in \Sigma^* \mid \psi(w) = \psi(v), v \in L\}$. Two languages are {\em letter-equivalent} if $\psi(L_1) = \psi(L_2)$.

A language $L$ is {\em semilinear} if $\psi(L)$ is a semilinear set. Equivalently, a language is semilinear if and only if it is letter-equivalent to some regular language \cite{harrison1978}. A family of languages is semilinear if all languages in it are semilinear, and it is said to be {\em effectively semilinear} if there is an algorithm to construct the constant and periods for each linear set from a representation of each language in the family.
 For example, it is well-known that all context-free languages are effectively semilinear \cite{Parikh}.

We will only define $\NCM$ and $\NPCM$ informally here, and refer to \cite{Ibarra1978} for a formal definition.
A one-way nondeterministic counter machine can be defined equivalently to a one-way nondeterministic pushdown automaton \cite{HU} with only a bottom-of-pushdown marker plus one other symbol. Hence, the machine can add to the counter (by pushing), subtract from the counter (by popping), and can detect emptiness and non-emptiness of the pushdown. A $k$-counter machine has $k$ independent counters. A $k$-counter machine $M$ is $l$-reversal-bounded, if $M$ makes at most $l$ changes between non-decreasing and non-increasing of each counter in every accepting computation. Let $\NCM$ be the class of one-way nondeterministic $l$-reversal-bounded $k$-counter machines, for some $k,l$ ($\DCM$ for deterministic machines). Let $\NPCM$ be the class of machines with one unrestricted pushdown plus some number of reversal-bounded counters. By a slight abuse of notation, we also use these names for the family of languages they accept.

Notation from AFL (abstract families of languages) theory is used from \cite{G75}.
A {\em full trio} is any family of languages closed under homomorphism, inverse homomorphism, and
intersection with regular languages. Furthermore, a {\em full AFL} is a full trio closed under union, concatenation, and Kleene-*.
Given a language family $\LL$, the smallest family containing $\LL$ that is closed
under arbitrary homomorphism is denoted by $\hat{\cal H}(\LL)$,
the smallest full trio containing
$\LL$ is denoted by $\hat{\cal T}(\LL)$, and the smallest
full AFL containing $\LL$ is denoted by $\hat{\cal F}(\LL)$.
Given  families $\LL_1$ and $\LL_2$, let
$\LL_1 \wedge \LL_2 = \{L_1 \cap L_2 \mid L_1 \in \LL_1, L_2 \in \LL_2\}$.
We denote by $\hat{\cal F}_{\NCM}(\LL)$ the smallest full AFL containing $\LL$ that is closed under intersection with languages from $\NCM$.

\section{Full AFLs Containing Counter Languages}

This section will start by showing that for every semilinear full trio $\LL$, the smallest full AFL containing $\LL$ that is also closed under intersection with $\NCM$ is a semilinear full AFL (Proposition \ref{fullAFLClosed}). First, an intermediate result is required.
\begin{proposition}
\label{prop1}
If $\cal L$  is a semilinear full trio, then
$\hat{\cal T}(\LL \wedge \NCM) = \hat{\cal H}(\LL \wedge \NCM)$ is a semilinear full trio.
\end{proposition}
\begin{proof}
Let ${\cal C} = \hat{\cal T}(\LL \wedge \NCM)$, and let 
$\hat{L} \in {\cal C}$ over alphabet $\Gamma = \{d_1, \ldots, d_s\}$. 
By definition, ${\cal C}$ is a full trio. It will be shown that $\hat{L}$ is semilinear.  Then $\hat{L}$ can be obtained from a language  $L$
in $\LL \wedge \NCM$ via a finite sequence of operations involving homomorphisms,
inverse homomorphisms, and intersections with regular sets.
Theorem 3.2.3 of \cite{G75} shows that for all non-empty
languages $L$, $\hat{\cal T}(L) = \{g_2(g_1^{-1}(L) \cap R_1) \mid R_1 \mbox{~is regular},
g_1, g_2 \mbox{~are decreasing homomorphisms}\}$. A homomorphism $g$
is decreasing if and only if $|g(a)| \leq 1$, for all letters $a$. Such homomorphisms
are called {\em weak codings}.
Hence, it is enough to consider that 
$\hat{L}$ is obtained from $L$ via the following sequence: an application of an inverse weak coding homomorphism
$g_1$, followed by an intersection with a regular language $R_1$, followed by an application of a weak coding homomorphism
$g_2$. Thus, $\hat{L} = g_2(g_1^{-1}(L) \cap R_1)$.
Since $L$ is in $\LL \wedge \NCM$,  there are $L_1 \in \LL$ and 
$L_2 \in \NCM$ such that $L = L_1 \cap L_2$. 
Let $L_2$ be accepted by a $k$-counter reversal-bounded $\NCM$ $M_2$, 
where, without loss of generality, all counters are $1$-reversal-bounded \cite{Ibarra1978},
all counters increase at least once, and all counters decrease to zero
before accepting.

Let $\Sigma = \{a_1, \ldots, a_n\}$ be the alphabet of  $L_1 \cup  L_2$, and so $L,L_1,L_2 \subseteq \Sigma^*$, $g_1$ is from $\bar{\Sigma}^*$ to $\Sigma^*$ for some alphabet $\bar{\Sigma}$, and so $g_1^{-1}(L) \subseteq \bar{\Sigma}^*, R_1 \subseteq \bar{\Sigma}^*$, and $g_2$ is from $\bar{\Sigma}^*$ to $\Gamma^*$. Introduce new symbols $\Delta = \{C_1, D_1, \ldots, C_k, D_k\}$ ($k$ is the number of counters).

Let $h_{\Sigma}$ be a homomorphism from $(\Sigma \cup \Delta)^*$ to $\Sigma^*$ 
that fixes each letter of $\Sigma$ and erases all letters of $\Delta$, and let $h_{\bar{\Sigma}}$ be a homomorphism $(\bar{\Sigma} \cup \Delta)^*$ to $\bar{\Sigma}^*$ that fixes each letter of $\bar{\Sigma}$ and erases all letters of $\Delta$.
There exists $R_2 \subseteq (\Sigma \cup \Delta)^*$, a regular set,
accepted by a nondeterministic finite automaton $M_2'$
that ``encodes'' the computation of $M_2$ (without doing the counting), as follows:
\begin{itemize}
\item $M_2'$ switches states as in $M_2$; $M_2'$ starts by simulating transitions on each counter being zero; 
\item every time $M_2$ adds to counter $i$, $M_2'$ instead reads
the input letter $C_i$; this is forced to happen at least once for each $i$, and after it reads the first $C_i$, it simulates transitions of $M_2$ where counter $i$ is positive;
\item every time $M_2$ subtracts from counter $i$,
$M_2'$ reads $D_i$; $M_2'$ verifies that at least one $D_i$ is read, and that no $C_i$ is read afterwards;
\item for each $i$, $1 \leq i \leq k$, at some nondeterministically guessed spot
after reading some $D_i$ symbol, 
$M_2'$ guesses that counter $i$ has hit zero,
and it no longer reads any $D_i$ symbol and simulates only transitions
on counter $i$ being zero;
\item $M_2'$ must end in a final state.
\end{itemize}

Let $L' =  h_{\Sigma}^{-1}(L_1)  \cap R_2$ (here, $h_{\Sigma}^{-1}(L_1)$ has symbols of $\Delta$ ``shuffled in'' to $L_1$). Let $L''$ be those words in $L'$
with the same number of $C_i$'s as $D_i$'s, for each $i$. Then it is evident that
$h_{\Sigma}(L'') = L$ as the $\NFA$ $M_2'$ that accepts $R_2$ is behaving like $M_2$ but without counting; however the counting is occurring using the intersection in $L''$, and then the counter
symbols of $\Delta$ are erased with $h_{\Sigma}$.
But looking only at $L'$, it must be that $L' \in \LL$ since $\LL$ is a full trio.
Let $\bar{g_1}$ be the extension of $g_1$ to be a homomorphism from $(\bar{\Sigma} \cup \Delta)^*$ to $(\Sigma \cup \Delta)^*$, where each letter of $\Delta$ is mapped to itself, and let $\bar{g_2}$ be the extension of $g_2$  to be a homomorphism from $(\bar{\Sigma} \cup \Delta)^*$ to $(\Gamma \cup \Delta)^*$,
where each letter of $\Delta$ is mapped
to itself. Let
$R'_1 = h_{\bar{\Sigma}}^{-1}(R_1)$; that is, it has symbols of $\Delta^*$ ``shuffled in''.
Note that $R'_1$ is regular since the regular languages are
closed under inverse homomorphism.

Then $L''' = \bar{g}_2 ( \bar{g_1}^{-1} (L' ) \cap R'_1)$, which is also in 
$\LL$, over $(\Gamma \cup \Delta)^*$. Note that since $g_1$ is a weak coding homomorphism, $\bar{g_1}$ is as well, and therefore $\bar{g_1}^{-1}$ simply operates as $g_1^{-1}$ does while fixing all letters of $\Delta$. If we then take $L'''$ and intersect it with all words where the number of $C_i$'s is equal to the number of $D_i$'s for each $i$, and then erase all $C_i$'s and $D_i$'s, we obtain $\hat{L}$. Let the Parikh image order the
letters of this alphabet 
$d_1, \ldots, d_s, C_1,D_1, \ldots, C_k,D_k$. 
This Parikh image of  $L'''$
gives a set 
$Q'''\subseteq \mathbb{N}_0^{s+2k}$, which is semilinear since
$\LL$ is semilinear.
Let $Q''''$ be the set obtained from $Q'''$ by
enforcing that the number of $C_i$'s is equal to $D_i$'s for each
$i$, which is also semilinear since the intersection of two semilinear sets is again semilinear \cite{Gins}.
Then $\bar{Q}$, the set obtained from $Q''''$ by projection on
the first $s$ coordinates, is also semilinear and 
this is the Parikh image of $\hat{L}$. Hence, $\hat{L}$ is semilinear.

By definition, $\hat{\cal H}(\LL \wedge \NCM) \subseteq \hat{\cal T}(\LL \wedge \NCM)$. 
To show $\hat{\cal T}(\LL \wedge \NCM) \subseteq \hat{\cal H}(\LL \wedge \NCM)$,
let $\hat{L} \in \hat{\cal T}(\LL \wedge \NCM)$. Using the proof that $\hat{\cal T}(\LL \wedge \NCM)$ is semilinear above, from $L''' \in \LL$,
it is possible to then intersect this language with an $\NCM$
that verifies that the number of $C_i$'s is equal to the number
of $D_i$'s, for each $i$. And then, a homomorphism that
erases elements of $\Delta$ can be applied to obtain $\hat{L}$.
\end{proof}

The next result is relatively straightforward from
results in \cite{G75,Ginsburg1971}, however we have
not seen it explicitly stated as we have done. 
From Corollary 2, Section 3.4 of \cite{G75}, for any full trio $\LL$, the smallest full AFL containing $\LL$ is the substitution of the
regular languages into $\LL$.
And from \cite{Ginsburg1971}, the substitution closure of one
semilinear family into another is semilinear. Therefore, we obtain:
\begin{lemma}
If $\LL$ is a semilinear full trio, then 
$\hat{{\cal F}}(\LL)$ is semilinear.
\label{lemma2}
\end{lemma}

For semilinear full trios $\LL$, $\hat{\cal T}(\LL \wedge \NCM)$ is a semilinear full trio
by Proposition \ref{prop1}, and starting with this family and applying Lemma \ref{lemma2},
the smallest full AFL containing intersections of languages in $\LL$ with $\NCM$ is semilinear.
\begin{proposition}   
\label{fullAFL}
If $\LL$ is a semilinear full trio, then $\hat{\cal F}(\LL \wedge \NCM)$ is semilinear.
\end{proposition}
It is worth noting that this procedure can be iterated, as
therefore $\hat{\cal F}(\hat{\cal F}(\LL \wedge \NCM) \wedge \NCM)$ must also be a semilinear
full AFL, etc.\ for additional levels. However, it is an interesting open question as to whether there is a strict hierarchy with respect to this iteration.

One could also consider the smallest full AFL containing $\LL$ that is closed under intersection with $\NCM$. Here, the intersections with $\NCM$ can occur arbitrarily many times, even after or in between applying the other full AFL operations.
\begin{proposition}
If $\LL$ is a semilinear full trio, then $\hat{\cal F}_{\NCM}(\LL)$ is semilinear.
\label{fullAFLClosed}
\end{proposition}
\begin{proof}
Let $\hat{L} \in  \hat{\cal F}_{\NCM}(\LL)$. Then $\hat{L}$ is obtained from some language $L \in \LL$ via some sequence of the full AFL operations, plus some number, $n$ say, of intersections with $\NCM$s. Hence,
$$\hat{L} \in {\cal C} = \overbrace{\hat{\cal F}(\hat{\cal F}( \cdots \hat{\cal F}}^{n}(\LL \wedge \overbrace{\NCM) \wedge \NCM) \wedge \cdots \wedge \NCM}^{n}).$$ By iterating Proposition \ref{fullAFL} $n$ times, ${\cal C}$ is semilinear, hence $\hat{L}$ is semilinear.
\end{proof}

In contrast, it is shown in \cite{G75} that 
$\hat{\cal H}(\hat{\cal F}(\{a^n b^n \mid n>0\}) \wedge \hat{\cal F}(\{a^n b^n \mid n>0\}))$ is equal to the family of recursively enumerable languages. 
Therefore, $\hat{\cal H}(\hat{\cal F}(\NCM) \wedge \hat{\cal F}(\NCM))$ is also equal to the family of recursively enumerable languages (which is not semilinear).
But in Propostion \ref{fullAFLClosed}, only intersections with languages in $\NCM$s are allowed, and not intersections with languages in $\hat{\cal F}(\NCM)$, thereby creating the large difference.

Many acceptor and grammar systems are known to be semilinear full trios, such as finite-index
$\ETOL$ systems \cite{RozenbergFiniteIndexETOL}, indexed grammars with a bound on the number of variables appearing in every sentential form (called uncontrolled finite-index) \cite{LATA2017journal}, multi-push-down machines (which have $k$ pushdowns that can simultaneously be written to, but they can only pop from the first non-empty pushdown) \cite{multipushdown}, a Turing machine variant with one finite-crossing worktape \cite{Harju2002278}, and pushdown machines that can flip their pushdown up to $k$ times \cite{Holzer2003}.
\begin{corollary}
Let $\LL$ be any of the following families:
\begin{itemize}
\item languages generated by context-free grammars,
\item languages generated by finite-index $\ETOL$, 
\item languages generated by uncontrolled finite-index indexed languages, 
\item languages accepted by one-way multi-push-down machine languages,
\item languages accepted by one-way read-only input nondeterministic Turing machines
with a two-way finite-crossing read/write worktape,
\item languages accepted by one-way $k$-flip pushdown automata.
\end{itemize}
Then $\hat{\cal F}_{\NCM}(\LL)$ is
a semilinear full AFL.
\end{corollary}
A simplified analogue to this result is known for certain types of
machines \cite{Harju2002278}, although the new result here
is defined entirely using closure properties rather than machines.
Furthermore, the results in \cite{Harju2002278}
do not allow Kleene-* type
closure as part of the full AFL properties. For the machine models ${\cal T}$ above, it is an easy
exercise to show that augmenting them with reversal-bounded counters to produce ${\cal T}_c$, the languages
accepted by ${\cal T}_c$ are a subset of the smallest full AFL closed under intersection with $\NCM$ containing languages in ${\cal T}$. Hence, these models augmented by counters only accept semilinear languages. Similarly, this type of technique also
works for grammar systems, as we will see in Section \ref{sec:indexed}.

In addition, in \cite{Ginsburg1971}, it was shown that
if $\LL$ is a semilinear family, then the smallest AFL containing the commutative closure of languages in $\LL$ is a semilinear AFL.
It is known that the commutative closure of every semilinear
language is in $\NCM$ \cite{eDCM}, and we know now 
that if we have a semilinear full trio $\LL$, 
then the smallest full AFL containing $\LL$
is also semilinear. So, we obtain an alternate proof that is an immediate corollary since we know that the smallest full AFL containing $\NCM$ is a semilinear full AFL.

For any semilinear full trio $\LL$ where the semilinearity and the intersection with regular language properties are effective, the membership and emptiness problems in $\LL$ are decidable. Indeed, to decide emptiness, it suffices to check if the semilinear set is empty. And to decide if a word $w$ is in $L$, one constructs the language $L \cap \{w\}$, then emptiness is decided. 
\begin{corollary}
For any semilinear full trio $\LL$ where the semilinearity and intersection with regular language properties are effective, then the membership, emptiness, and infiniteness problems are decidable for languages in $\hat{\cal F}_{\NCM}(\LL)$. In these cases,
$\hat{\cal F}_{\NCM}(\LL)$ are a proper subset of the recursive languages.
\end{corollary}
As membership is decidable, the family must only contain recursive languages, and the inclusion must be strict as the recursive languages are not closed under homomorphism.

As another consequence, we provide an interesting decomposition theorem of semilinear languages into linear parts.
Consider any semilinear language $L$, where its Parikh image is a finite union of linear sets
$A_1, \ldots, A_k$, and the constant and periods for each linear set can be constructed. Then we
can effectively create languages
in perhaps another semilinear full trio separately accepting those words in 
$L_i = \{w \in L \mid \psi(w) \in A_i\}$, for each $1 \leq i \leq k$.
\begin{proposition}
Let $\LL$ be a semilinear full trio, where semilinearity is effective. Given
$L \in \LL$, we can determine representations of disjoint
simple sets (ie.\ disjoint linear sets where the periods form a basis) $A_1, \ldots, A_k$
such that the Parikh image
of L is $A = A_1\cup \cdots \cup A_k$,
and $L_i = \{ w \in L \mid \psi(w) \in A_i \} \in \hat{\cal F}_{\NCM}(\LL)$, for $1 \leq i \leq k$.
\end{proposition}
\begin{proof}
Since semilinearity is effective, we can construct a representation of
linear sets $A_1 , \ldots, A_k$. 
Moreover, it is known that given any set of constants and periods generating a semilinear set $Q$, 
it is possible to effectively construct another set of constants and periods that forms a disjoint finite union of simple sets
also generating $Q$ \cite{Flavio,Sakarovitch}. Therefore, we can assume
$A_1, \ldots, A_k$ are of this form.
An $\NCM$ $M_i$ can be created to
accept $\psi^{-1}(A_i)$, for each $i$, $1 \leq i \leq k$ as follows:
if $L\subseteq \{a_1, \ldots, a_n\}^*$, then $M_i$ has $n$ counters. If
$(x_1, \ldots, x_n)$ is the constant of $A_i$, then $M_i$ adds $x_j$ to counter $j$ for each $j$. Then, for each period, $(y_1, \ldots, y_n)$, $M_i$ nondeterministically guesses some number $c$ and adds $cy_j$ to counter $j$ for each $j$. At this point, the counters can contain any value from $A_i$. From here, for every $a_j$ read as input, $M_i$ subtracts one from counter $j$, and accepts at the end of the input if all counters are empty.
Hence, 
$L_i = L \cap L(M_i) \in  \hat{\cal F}_{\NCM}(\LL)$, for each $i$, $1 \leq i \leq k$.
\end{proof}

Therefore, by moving to a more general full trio (contained in the recursive
languages), it is possible to decompose
a language into separate (disjoint) languages such that each has
one of the linear sets as its Parikh image. 

\section{Application to General Multi-Store Machine Models}

In \cite{G75}, a generalized type of multitape automata was studied,
called multitape abstract families of automata (multitape AFAs).
We will not define the notation used there, but in Theorem 4.6.1 (and
Exercise 4.6.3), it is shown that if we have two types of automata
${\cal M}_1$ and ${\cal M}_2$ (defined using the AFA formalism), accepting language
families $\LL_1$ and $\LL_2$ respectively, then the languages accepted by automata combining
together the stores of ${\cal M}_1$ and ${\cal M}_2$, accepts exactly the family
$\hat{{\cal H}}(\LL_1 \wedge \LL_2)$.
This is shown for machines accepting full AFLs in Theorem 4.6.1 of \cite{G75}, and for union-closed full trios mentioned in
Exercise 4.6.3. We will show that this is tightly coupled with this precise definition of AFAs,
as we will define two simple types of automata where each on their own accept a semilinear family, but combining the two stores together to form one multitape model accepts non-semilinear languages.

Given a family of one-way acceptors ${\cal M}$, let ${\cal M}_c$ be those acceptors augmented by reversal-bounded counters.
A checking stack automaton ($\NCSA$) $M$ is a 
one-way $\NFA$
with a store tape, called a stack. At each
move, $M$  pushes a string (possibly $\lambda$) on the
stack, but $M$ cannot pop.
And, $M$ can enter and read from the inside of the stack in a two-way read-only fashion. But once the machine enters
the stack, it can no longer change the contents. The checking stack automaton is said to be {\em restricted} (or {\em no-read}
using the terminology of \cite{DLT2017TCS}), if it does not read from the inside of the stack until the end of the input.
We denote by $\RNCSA$ the family of machines, as well as the
family of languages described by the machines above.
A preliminary investigation of $\RNCSA_c$ was done in \cite{DLT2017TCS}.

Here, we will show the following:
\begin{enumerate}
\item $\RNCSA$ is a full trio of semilinear languages equal to the regular languages,
\item $\hat{{\cal F}}(\RNCSA \wedge \NCM)$ and $\hat{{\cal F}}_{\NCM}(\RNCSA )$
are semilinear full AFLs,
\item every language in $\RNCSA \wedge \NCM$ is accepted by some machine in
$\RNCSA_c$,
\item there are non-semilinear languages accepted by machines
in $\RNCSA_c$.
\end{enumerate}
Therefore, $\RNCSA_c$ contains some languages not in the smallest full AFL containing $\RNCSA$ closed under intersection with $\NCM$,
and the multitape automata and results from \cite{G75} and \cite{Harju2002278} do not apply to this type of automata.

\begin{proposition} \label{prop-reg}
$\RNCSA$ accepts exactly the regular languages, which is a
full trio of semilinear languages.
\end{proposition}
\begin{proof}
It is clear that all regular languages are in $\RNCSA$.
For the other direction, take an $\RNCSA$ machine $M$, and assume
without loss of generality that the input alphabet $\Sigma$ and
the stack alphabet $\Gamma$ are disjoint. Construct
a two-way NFA ($2\NFA$) $M'$
over $(\Sigma \cup \Gamma)^*$ whose input is divided into segments
$u_1 v_1 \cdots u_n v_n$, where  $u_i \in (\Sigma \cup \{\lambda\})$
and  $v_i \in \Gamma^*$.
$M'$ simulates $M$ by first verifying that $M$, when reading $u_i$,
writes  $v_i$ on the stack. When $M'$ simulates the two-way read-only phase of $M$
(which only occurs in $M$ after reaching the end of the input),
it does so by using the two-way $\NFA$ and skipping over the segments of $\Sigma$.
Since this language accepted by the $2\NFA$ $M'$ is regular,
the language obtained by erasing all letters of $\Gamma$ via
homomorphism is also regular, which is exactly $L(M)$.

\end{proof}

From Proposition \ref{fullAFL}, the following is true:
\begin{corollary}
$\hat{{\cal F}}(\RNCSA \wedge \NCM) =\hat{{\cal F}}_{\NCM}(\RNCSA) = \hat{{\cal F}}(\NCM)$
is a semilinear full AFL.
\end{corollary}

Since $\RNCSA$ is equal to the family of regular languages, and $\NCM$ is
closed under intersection with regular languages, the following
is true:
\begin{proposition}
$\hat{{\cal F}}_{\NCM}(\RNCSA) = \NCM \subsetneq  \RNCSA_c$.
Furthermore, the latter family contains non-semilinear languages.
\label{non-semi}
\end{proposition}
\begin{proof}
Containment is immediate since $\RNCSA_c$ has reversal-bounded counters.
The non-semilinear 
$L = \{a^i b^j \mid i, j \ge 1, j \mbox{~is divisible by~} i\}$
can be accepted by an $\RDCSA_c$ $M$ with one counter that makes
only one reversal.
$M$, on input $x$ 
checks that $x = a^ib^j$ for some $i, j \ge 1$, copies $a^i$ onto the stack, and
increments the counter to $j$. Then $M$ makes multiple 
left-to-right and right-to-left sweeps on $a^i$ with the stack while in parallel
decrementing the counter to check that $j$ is divisible by $i$.
\end{proof}

It is concluded that $\RNCSA_c$ contains some languages
not in $\hat{{\cal F}}_{\NCM}(\RNCSA) = \NCM$, since $\NCM$ is semilinear. Then it is clear that combining together the stores
of $\RNCSA$ and $\NCM$ accepts significantly more than
$\hat{{\cal H}}(\RNCSA \wedge \NCM)$ as is the case for
multitape AFA \cite{G75}.
The reason for the discrepancy between
this result and Ginsburg's result is that the definition
of multitape AFA allows for reading the input while performing 
instructions (like operating in two-way read-only mode in the stack).
In contrast, $\RNCSA$ does not allow this behavior. And if
this behavior is added into the definition, the full capability
of checking stack automata is achieved which accepts non-semilinear languages.

A similar analysis can be done using the method developed in \cite{Harju2002278} for augmenting the machine models
with counters.
Let $\cal M$ be a family of  one-way acceptors with some type of 
store structure $X$.
For example, if the storage $X$  is a pushdown stack, then $\cal M$ is
the family
of nondeterministic pushdown automata ($\NPDA$s). 
In \cite{Harju2002278},
the following was shown for many families $\cal M$:

\begin{enumerate}
\item[(*)]
If $\cal M$ is a semilinear family (i.e, the languages
accepted by the machines in $\cal M$ have a semilinear Parikh image), then
${\cal M}_c$ is also a semilinear family.
\end{enumerate}
It was not clear in \cite{Harju2002278} whether the result above is true for all types of one-way
acceptors, in general or for which types (*) holds. However, the family $\RNCSA$ (equal to the regular languages) is semilinear
(Proposition \ref{prop-reg}), but $\RDCSA_c$ is not semilinear
(Proposition \ref{non-semi}).

\section{Properties of semilinear language families}
\label{closure}

This section investigates certain properties of semilinear language families.

\begin{definition}
Given a language family $\LL$, define the following families:
\begin{eqnarray*}
\overline{\LL} &=& \{ \overline{L}  \mid  L \in \LL\},\\
\LL_D &=& \{L_1 - L_2  \mid  L_1, L_2 \in \LL\},\\
\LL_{\cup} &=& \{L_1 \cup L_2 \mid L_1, L_2 \in\LL\},\\
\LL_{\cap} &=& \{L_1 \cap L_2 \mid L_1, L_2 \in \LL\},\\
\LL\LL &=& \{L_1L_2  \mid  L_1,L_2 \in \LL\},\\
\LL^* &=& \{L^* \mid L \in \LL\},\\
\LL_{RQ} & =& \{L_1L_2^{-1}  \mid L_1, L_2 \in \LL\}, \mbox{~(right quotient)},\\
\LL_{LQ} & =& \{L_1^{-1}L_2  \mid L_1, L_2 \in \LL\}, \mbox{~(left quotient)},\\
{\cal H}(\LL) &=& \{h(L) \mid L \in \LL, h \mbox{~a homomorphism}\},\\
{\cal H}^{-1}(\LL) &=& \{h^{-1}(L) \mid L \in \LL, h \mbox{~a homomorphism}\}.
\end{eqnarray*}
\end{definition}

If $\LL$ is semilinear, an interesting question is whether the defined families above
must also be semilinear.
In \cite{Ginsburg1971}, it is shown that the substitution of one semilinear family into another is again semilinear. This immediately implies that if $\LL$ is a semilinear family, then all of $\LL^*, \LL_{\cup},\LL\LL$, and
${\cal H}(\LL)$ are also semilinear.
For the remaining properties, we have not seen proofs in the literature, and therefore include short proofs here.

\begin{proposition}  If $\LL$ is semilinear,  then all of the following need not be semilinear: $\overline{\LL}, \LL_D, \LL_{RQ}, \LL_{LQ}, \LL_{\cap}, {\cal H}^{-1}(\LL)$.
\end{proposition}
\begin{proof}
First, it will be shown for $\overline{\LL}$.
Let  $L = \{a^1 \# a^2 \# \cdots \# a^k \#  \mid  k \ge 1 \}$ where $a$ is a letter.
Then the complement of $L$, $\overline{L}$ can easily be accepted  by
an $\NCM$ with one 1-reversal counter which, when given
an input $w$,  nondeterministically selects (1) or (2) below:
\begin{enumerate}
\item accepts, if  $w$ is not in a valid format, i.e., not of the form $(a^+ \#)^+$.
(M does not need the counter.)
\item accepts $w$ if it is of the form $a^{i_1} \# \cdots \# a^{i_k}$  but $i_r +1 \ne i_{r+1}$
for some $r$.  ($M$ uses a 1-reversal counter.)
\end{enumerate}
Since all $\NCM$ languages are semilinear, $\overline{L}$ is semilinear, but $L$ is not semilinear (if it were semilinear, then projecting onto $a$ would be semilinear, but all unary semilinear languages are regular \cite{harrison1978} and this language is not regular by the pumping lemma).
This also implies non-closure for $\LL_D$.

Next, for right quotient, it is known that there is a non-recursively enumerable unary language $L \subseteq a^*$ (that is not semilinear) \cite{Minsky}. Let $L' = cL d \cup d a^* c$. Then $L'$ is semilinear since it has the same Parikh image as the regular language $da^*c$. But the right quotient of $L'$ with $d$ is $cL$, which is not semilinear. 

Similarly, the left quotient of $L'$ with $c$ is $Ld$, which is not semilinear. The result for intersection is also similar.

For inverse homomorphism, take a homomorphism $h$ that maps $b$ to $ca$, $e$ to $a$, and $f$ to $d$, and $g$ to $ac$. Then
$h^{-1} (L') = b L'' f \cup f a^* g$ where $L'' = L a^{-1}$. The language $h^{-1}(L')$ is clearly not semilinear.

\end{proof}

In contrast, it can be seen that for inverse homomorphisms where the homomorphisms are {\em weak codings} (that is, $|h(a)| \leq 1$ for all $a \in \Sigma$), then the resulting family
is semilinear, as inverse homomorphisms act just as substitutions (as mentioned, the
substitution closure of a semilinear family is semilinear) with the additional arbitrary insertion of characters erased by $h$ (which can be added in by placing another period in each linear set of the semilinear set with all $0$'s except for a $1$ for the position of the character erased by the homomorphism).

These closure properties will motivate the next notion that can help define ``well-behaved'' semilinear
languages.
\begin{definition}
A semilinear language $L$ is {\em well-behaved} if $\hat{\cal T}(L)$ is semilinear;
that is, it is well-behaved if closing it under all full trio operations only give
semilinear languages.
\end{definition}

Some basic facts are in order.
\begin{proposition}
The following are true:
\begin{itemize}
\item if $L \in \LL$, a semilinear full trio, then $L$ is well-behaved,
\item not all semilinear languages are well-behaved.
\end{itemize}
\end{proposition}
\begin{proof}
The first property is immediate since $\hat{\cal T}(L) \subseteq \LL$.

Consider a non-recursively enumerable unary language $L \subseteq a^*$. Therefore,
this language is not semilinear (all unary semilinear languages are regular).
Consider $L' = b L c \cup c a^* b$. Then $L'$ is semilinear since it has
the same Parikh image as the regular language $c a^* b$. But $L' \cap b a^* c = bLc$,
which is not semilinear. Thus, the closure of $L'$ by intersection with regular languages
gives non-semilinear languages.\\

\end{proof}

Consider the following language family:
$$\LL_{WB} = \{L \mid L \mbox{~is well-behaved.}\}.$$

\begin{proposition}
$\LL_{WB}$ is the largest semilinear full trio. That is, all semilinear full trios are contained in $\LL_{WB}$.
\end{proposition}
\begin{proof}
First, it is a semilinear full trio since all languages in it are semilinear,
and the closure of each under the full trio properties are in it.

Furthermore, it is the largest since any language not in it must either not be semilinear,
or closing it under the full trio operations produces languages that are not semilinear.

\end{proof}
This is similar to the known result that there is a largest semilinear AFL
\cite{Ginsburg1971}.
This is an interesting language family, as properties that hold for this single
language family hold for all semilinear full trios.

For example, consider the following. A bounded language $L \subseteq w_1^* \cdots w_k^*$ is called
{\em bounded Ginsburg semilinear} (often just called bounded semilinear in the literature) if
the set $\{(i_1, \ldots, i_k) \mid w_1^{i_1} \cdots w_k^{i_k} \in L\}$ is a semilinear set.
The following is true from \cite{CIAA2016}:
\begin{proposition}
All bounded languages in $\LL_{WB}$ are bounded Ginsburg semilinear languages.
\end{proposition}

Next, it follows from Theorem 3.2.3 of \cite{G75}, that for all non-empty
languages $L$, $\hat{\cal T}(L) = \{h_2(h_1^{-1}(L) \cap R) \mid R \mbox{~is regular},
h_1, h_2 \mbox{~are decreasing homomorphisms}\}$. The homomorphisms
are both weak codings. Also, semilinear languages are
closed under homomorphisms. Hence, the following is true:
\begin{proposition}
$L$ is well-behaved if and only if the family $$\{h^{-1}(L) \cap R \mid R \mbox{~is regular},
h \mbox{~is a weak coding homomorphism}\}$$ are all semilinear.
\end{proposition}
It is evident that $h^{-1}(L)$ must be semilinear
since it was previously noted that the family obtained from any semilinear family via inverse weak coding homomorphisms must be semilinear.
Hence, if one examines the family of semilinear languages
${\cal L} = \{h^{-1}(L) \mid h \mbox{~is a weak coding homomorphism}\}$,
then $L$ is well-behaved if and only if ${\cal L} \wedge {\cal L}(\NFA)$.

\section{Applications to Indexed Grammars with Counters}
\label{sec:indexed}

In this section, we describe some new types of grammars obtained from existing grammars 
generating a semilinear language family $\LL$, by adding counters. The languages generated by these new grammars are then shown to be contained in $\hat{\cal F}(\LL \wedge \NCM)$,
and by an application of Proposition \ref{fullAFL}, are all semilinear with positive decidability properties.

We need the definition of an indexed grammar introduced in  \cite{A}  by following the notation of \cite{HU}, Section 14.3.
 \begin{definition}\label{def0}
  An {\em indexed grammar} is a 5-tuple $G = (V, \Sigma, I, P, S)$, where 
$V, \Sigma, I$ are finite pairwise disjoint sets: the set of nonterminals, terminals, and indices, respectively, $S$ is the start nonterminal, and
$P$ is a finite set of  productions, each of the form either
$$
\mbox{\bf 1)\ } A \rightarrow \nu, \quad \mbox{\bf 2)\ } A \rightarrow Bf, \quad \mbox{or}\quad 
\mbox{\bf 3)\ } Af \rightarrow \nu,$$
where $A, B\in V,$ $f\in I$ and $\nu \in (V\cup \Sigma)^*$.
\end{definition}
Let $\nu$ be an arbitrary sentential form of $G$, which is of the form
 $$\nu = u_1 A_1 \alpha_1 u_2 A_2 \alpha_2 \cdots u_k A_k \alpha_k u_ {k+1},$$ 
 where $A_i\in V, \alpha _i \in I^*, u_i\in \Sigma^*, 1 \leq i \leq k, u_{k+1} \in \Sigma^*$.
 For a sentential form $\nu ' \in  (VI^* \cup \Sigma)^*$, we write $\nu \Rightarrow _G \nu '$ if one of the following three conditions holds:
\begin{enumerate}
\item There exists a production in $P$ of the form (1) $A \rightarrow w_1 C_1  \cdots w_{\ell} C_\ell w_ {\ell+1}$, $C_j \in V,
w_j \in \Sigma^*$,  and there exists $i$ with $1\leq i\leq k,$  $A_i =A$ and 
$$\nu' = u_1 A_1 \alpha_1 \cdots u_{i} (w_1 C_1\alpha_i  \cdots w_{\ell} C_\ell \alpha_iw_ {\ell+1}) u_{i+1} A_{i+1}\alpha_{i+1} \cdots u_k A_k \alpha_k u_ {k+1}.$$

\item There exists a production in $P$  of the form (2)  $A\rightarrow Bf$ and there exists $i$, $1\leq i\leq k,$ $A_i =A$  and 
$\nu '  = u_1 A_1 \alpha_1 \cdots u_{i} (B f \alpha_i) u_{i+1} A_{i+1}\alpha_{i+1} \cdots u_k A_k \alpha_k u_ {k+1}.$

\item There exists a production in $P$ of the form (3) $A f\rightarrow w_1 C_1  \cdots w_{\ell} C_\ell w_ {\ell + 1}$, $C_j \in V,
w_j \in \Sigma^*$,  
and an $i$, $1\leq i\leq k,$ $A_i =A$, $\alpha_i = f\alpha'_i, \alpha'_i\in I^*$,  with
$\nu  ' = u_1 A_1 \alpha_1 \cdots u_{i} (w_1 C_1\alpha'_i  \cdots w_{\ell} C_\ell \alpha'_iw_ {\ell+1}) u_{i+1} A_{i+1}\alpha_{i+1}  \cdots u_k A_k \alpha_k u_ {k+1}.$
 \end{enumerate}
Then, $\Rightarrow_G ^*$ denotes the reflexive and transitive 
closure of $\Rightarrow_G$. The language $L(G)$ generated by  $G$ is the set $L(G) = \{u\in \Sigma^* ~|~ S \Rightarrow_G ^* u\}.$

This type of grammar can be generalized to include monotonic counters as follows:
\begin{definition}
An {\em indexed grammar with $k$ counters} is defined as in indexed grammars, except
where rules (1), (2), (3) above are modified so that a rule
$\alpha \rightarrow \beta$ 
now becomes:
\begin{equation}
\alpha \rightarrow ( \beta,  c_1, \ldots , c_k),
\label{newprod}
\end{equation}
where $c_i \ge 0$, $1 \leq i \leq k$. Sentential forms are of the form
$(\nu, n_1, \ldots, n_k)$, and $\Rightarrow_G$ operates as do indexed grammars on $\nu$, and for a production in Equation \ref{newprod}, adds $c_i$ to $n_i$, for $1 \leq i \leq k$.
The language generated by $G$ with terminal alphabet $\Sigma$ and
start nonterminal $S$ is,
$L(G) = \{ u ~|~ u \in \Sigma^*,
(S, 0, \ldots , 0) \Rightarrow_G^* (u, n_1, \ldots, n_k), n_1 = \cdots = n_k\}$.

Given an indexed grammar with counters, the {\em underlying grammar} is the indexed grammar obtained by removing the counter components from productions.
\end{definition}

Although indexed grammars generate non-semilinear languages, restrictions will be studied that only generate semilinear languages.


An indexed grammar $G$ is  {\em linear} \cite{DP} if  the right side of every production of $G$ has at most one variable. 
Furthermore, $G$ is {\em right linear} if it is linear, and terminals can only appear to the left of a nonterminal in productions.
Let $\ILIN$ be the family of languages generated by linear indexed grammars, and let $\IRLIN$ be the family of languages
generated by right linear indexed grammars.

Similarly, indexed grammars with counters can be restricted to be linear.
An indexed grammar with $k$-counters is said to be {\em linear indexed} (resp.\ {\em right linear}) with $k$ counters, if the underlying grammar is linear (resp.\ right linear).
Let $\ILIN_c$ (resp.\ $\IRLIN_c$) be the family of languages generated by linear (resp.\ right linear) indexed grammars with counters.

\begin{example}
Consider 
the language $L = \{ v\$w  \mid v,  w \in \{a,b,c\}^*,
|v|_a = |v|_b = |v|_c, |w|_a = |w|_b = |w|_c \}$ which can be generated by a linear indexed
grammar with counters $G = (V,\Sigma,I,P,S)$ where
$P$ contains 
\begin{tabbing}
$S \rightarrow (S,1,1,1,0,0,0) \mid (S,0,0,0,1,1,1) \mid (T,0,0,0,0,0,0)$\\
$T \rightarrow (aT,1,0,0,0,0,0) \mid (bT,0,1,0,0,0,0) \mid (cT,0,0,1,0,0,0) \mid (\$R, 0,0,0,0,0,0)$\\
$R \rightarrow (aR,0,0,0,1,0,0) \mid (bR,0,0,0,0,1,0) \mid (cR,0,0,0,0,0,1) \mid (\lambda,0,0,0,0,0,0)$.
\end{tabbing}

This language cannot be generated by a linear indexed grammar \cite{LATA2017journal}.
\end{example}

Next, a characterization of languages generated by these grammars will be given with a sequence of results used towards the proof of Proposition 
\ref{thm1}.

In the following, $\Sigma$ is a terminal alphabet,
$C = \{c_1, \ldots , c_k \}$ (for some $k \ge 1$) is an alphabet
distinct from $\Sigma$, and $h_c$ is a homomorphism 
on $\Sigma \cup C$ defined by
$h_c(a) = a$ for each $a$ in $\Sigma$, and
$h_c(c_i) = \lambda$ for each $c_i$ in $C$.

\begin{lemma} \label{lem0}
If $L$ is in $\NCM$ (resp., $\NPCM$),
there is  regular language (resp., $\NPDA$) $R$ over the alphabet $\Sigma \cup C$
such that $L = h_c(\{w  ~|~ w \in R,  |w|_{c_1} = \cdots = |w|_{c_k}\})$.
\end{lemma}
\begin{proof}
Let $L \subseteq  \Sigma^*$ be accepted by an $\NCM$ (resp., $\NPCM$)
with $n$ 1-reversal
counters. Let $C = \{b_1, c_1, \ldots, b_n,c_n\}$ be 
an alphabet distinct from $\Sigma$.  (Thus $k = 2n$.)
It follows from the constructions in \cite{Ibarra1978}, that
there is a regular language $R$ (resp., $\NPDA$) over alphabet
$\Sigma \cup C$  such that 
$L = h_c(\{w  ~|~ w \in R,  |w|_{b_1} = |w|_{c_1}, \ldots,
|w|_{b_n} = |w|_{c_n}\})$.
Now let $R' = (b_1c_1)^* \cdots (b_nc_n)^*R$.  Clearly, $R'$ is also
regular (resp., $\NPDA$), and 
$L = h_c(\{w  ~|~ w \in R',  |w|_{b_1} = |w|_{c_1} = \cdots = 
|w|_{b_n} = |w|_{c_n}\})$.

\end{proof}

\begin{lemma} \label{lem1}
If $L_1 \subseteq \Sigma^*$ is in $\ILIN$, and
$L_2 \subseteq \Sigma^*$ is in $\NCM$,
then $L_1 \cap L_2 \in \ILIN_c$.
\end{lemma}
\begin{proof}
By Lemma \ref{lem0}, since $L_2$ is in $\NCM$,
there is  regular set $R$ over alphabet $\Sigma \cup C$
such that $L_2 = h_c(\{w  ~|~ w \in R,  |w|_{c_1} = \cdots = |w|_{c_k}\})$.
Also, $L_1' = h_c^{-1}(L_1)$ is also a linear indexed language
since the family is a full trio \cite{DP}, and $L_3 = L_1' \cap R$ is 
also a linear indexed language.  Let $G_3$
be a linear indexed grammar generating $L_3$.

We can now construct from $G_3$ a linear indexed grammar with counters
generating $L = L_1 \cap L_2$, such that,
if $\alpha \rightarrow \beta$ is a production in $G_3$,
then
$\alpha \rightarrow (h_c(\beta), |\beta|_{c_1}, \ldots, |\beta|_{c_k})$ is a production in $G_4$.
Then $L(G_4) = L_1 \cap L_2$.

\end{proof}

Since languages generated by linear indexed grammars with counters are clearly
closed under homomorphism, the following is true:
\begin{corollary} \label{cor1}
Let $h$ be a homomorphism, $L_1 \in \ILIN$, and $L_2 \in \NCM$.  
Then $h(L_1 \cap L_2) \in \ILIN_c$.
\end{corollary}

\begin{lemma} \label{lem2}
If $L \in \ILIN_c$,
then $L = h(L_1 \cap L_2)$
for some homomorphism $h$, $L_1 \in \ILIN$, and 
$L_2 \in \NCM$.
\end{lemma}
\begin{proof}
Let $L$ be generated by $G$.  Construct a linear indexed grammar
(without counters) $G_1$ as follows:

If $\alpha \rightarrow (\beta, d_1, \ldots, d_k)$ is a rule in $G$,
then
$\alpha \rightarrow \beta'$ is a rule in $G_1$,
where $\beta' = c_1^{d_1} \cdots c_k^{d_k} \beta$,
(i.e., we append to the left of $\beta$ a terminal string representing
the increments in the counters).

Let $L_1$ be the language generated by $G_1$.
Let $L_2 = \{ w  ~|~ w \in (\Sigma \cup C)^*,
|w|_{c_1} = \cdots = |w|_{c_k} \}$.  Clearly $L_2$ is an $\NCM$ language,
and $L = h_c(L_1 \cap L_2)$

\end{proof}


\begin{proposition} \label{thm1}
$L \in \ILIN_c$ 
if and only if there is a homomorphism $h$, $L_1 \in \ILIN$,
and $L_2 \in \NCM$ such that $L = h(L_1 \cap L_2)$.
\end{proposition}
\begin{proof}
This follows immediately from Corollary \ref{cor1} and Lemma \ref{lem2}.

\end{proof}

Implied from the above result and Proposition \ref{fullAFL} and since $\ILIN$ is an effectively semilinear trio \cite{DP} is that
$\ILIN_c \subseteq \hat{\cal F}(\ILIN \wedge \NCM)$, and therefore $\ILIN_c$ is effectively semilinear.

\begin{corollary}
The languages generated by linear indexed grammar with counters are effectively semilinear, with decidable emptiness, membership, and infiniteness problems.
\end{corollary}


Next, a machine model characterization of right linear indexed grammars with counters will be provided. Recall that
an $\NPCM$ is a pushdown automaton augmented by reversal-bounded counters.
The proof uses 
the fact that every context-free language
can be generated by a right-linear indexed grammar \cite{DP}.
%


%


\begin{proposition} \label{thm2}
$\IRLIN_c = \NPCM$.
\end{proposition}
\begin{proof}
First, it will be show that $\NPCM \subseteq \IRLIN_c$.
Let $L \in \NPCM$. Then, by Lemma \ref{lem0}, there is an $\NPDA$ $L_1$
over $\Sigma \cup C$
such that $L = h_c(\{w ~|~ w \in L_1, |w|_{c_1} = \cdots = |w|_{c_k} \})$.
It is known that every context-free language can be generated by a right-linear indexed grammar \cite{DP}, and hence
there is a right-linear indexed grammar $G_1$ generating $L_1$.
Construct from $G_1$, a right-linear indexed grammar $G$ with counters
generating $L$, such that, 
if $\alpha \rightarrow \beta$ is a production in $G_1$,
then
$\alpha \rightarrow (h_c(\beta), |\beta|_{c_1}, \ldots, |\beta|_{c_k})$ is a production in $G$. Then $L(G) = L$.

Next, the converse will be shown.
Let $G$ be a right-linear indexed grammar with counters.
We first construct a right-linear grammar (without counters)
$G_1$ generating a language $L_1$ as in the proof of Lemma \ref{lem2}.
Then $L(G_1)$ is a context-free language, and can be accepted by an $\NPDA$ $M_1$.
We then construct an $\NPCM$ $M$ which, when given an input 
$w$, simulates $M_1$ and checks that $|w|_{c_1} = \cdots = |w|_{c_k}$
using 1-reversal counters.  Finally, we construct from $M$ another
$\NPCM$ $M'$ which erases the $c_i$'s.  Clearly, $L(M') = L$.

\end{proof}

We conjecture that the family of languages generated by right-linear
indexed grammars with counters (the family of $\NPCM$ languages)
is properly contained in the family of languages generated by
linear indexed grammars with counters. Candidate witness languages
are $L = \{ w\$w  ~|~  w \in \{a,b,c\}^* , |w|_a + |w|_b = |w|_c \}$ and
$L' = \{ w\$w  ~|~  w \in \{a,b\}^* \}$. It is known that $L'$ is generated
by a linear indexed grammar \cite{DP}, and hence $L$ can be generated by such a grammar
with two counters. But, both $L'$ and $L$ seem unlikely to be
accepted by any $\NPCM$. 
Therefore, indexed grammars with counters form quite a general semilinear family as it seems likely to be more
general than $\NPCM$.


Next, another subfamily of indexed languages is studied that are even more
expressive than linear indexed grammars but still only generate semilinear languages.

An indexed grammar  $G = (V,\Sigma,I,P,S)$ is said to be {\em uncontrolled} index-$r$ if,
every sentential form in every successful derivation has at most $r$ nonterminals.
$G$ is
uncontrolled finite-index if $G$ is uncontrolled index-$r$, for some $r$.
Let $\UFIN$ be the languages generated by uncontrolled finite-index indexed grammars.

Uncontrolled finite-index indexed grammars have also been studied under the name of  
breadth-bounded indexed grammars in \cite{ICALP2015,LATA2017journal}, where 
it was shown that the languages generated by these grammars are a semilinear full trio.

This concept can then be carried over to indexed grammars with counters.
\begin{definition}
An indexed grammar with $k$-counters is {\em uncontrolled index-$r$} (resp.\ {\em uncontrolled finite-index}) if the
underlying grammar is uncontrolled index-$r$ (resp.\ uncontrolled finite-index). 
Let $\UFIN_c$ be the languages generated by uncontrolled finite-index indexed grammar with $k$-counters, for some $k$.
\end{definition}

One can easily verify that 
Proposition \ref{thm1} also
applies to
uncontrolled finite-index indexed grammars with counters.
Hence, we have:
\begin{proposition} 
$L \in \UFIN_c$ if and only if there is a homomorphism $h$, $L_1 \in \UFIN$,
$L_2 \in \NCM$ such that $L = h(L_1 \cap L_2)$.
\label{thm3}
\end{proposition}

Implied from the above proposition and Proposition \ref{fullAFL} also is that these new languages are all semilinear.
\begin{corollary}
$\UFIN_c$ is effectively semilinear, with decidable emptiness, membership, and infiniteness problems.
\end{corollary}

Hence, $\IRLIN_c \subseteq \ILIN_c \subseteq \UFIN_c$. We conjecture that both containments are strict; the first was discussed previously,
and the second is likely true since $\ILIN \subsetneq \UFIN$ \cite{LATA2017journal}. Hence, $\UFIN_c$ forms quite a general
semilinear family, containing $\NPCM$ with positive decidability properties.

\section{Conclusions and Future Directions}
It has been previously shown that certain types of machine models accepting only semilinear languages can be augmented by reversal-bounded counters to create a more general machine model, while maintaining semilinearity and positive decision properties. However, this approach did not clearly define what types of models would work with this augmentation, and it did not work with other mechanisms for describing languages. Here, a closure property theoretic method is developed, and it is shown that, for every semilinear full trio $\LL$, the smallest full AFL containing $\LL$ also closed under intersection with reversal-bounded multicounter languages ($\NCM$) is semilinear. Furthermore, the semilinearity is effective in the resulting family if it is effective (with other properties) in $\LL$. 

This can be applied in numerous ways. For example, it is shown that if certain subclasses of indexed grammars (linear indexed, or uncontrolled finite-index) are augmented by counters with additional components of the grammars that function like counters, then the resulting families are more general, yet they remain semilinear and have decidable emptiness, membership, and infiniteness problems.
There are also other applications, such as to analyzing definitions of abstract automata with multitape stores.

Several open problems remain. It is open whether the application of Proposition 
\ref{fullAFL} creates a strict hierarchy. With respect to indexed grammars with counters, it is open as to whether right-linear grammars are strictly weaker than linear grammars, and whether those are weaker than uncontrolled finite-index grammars.

\section*{Acknowledgments}
The research of O. H. Ibarra was supported, in part, by
NSF Grant CCF-1117708. The research of I. McQuillan was supported, in part, by Natural Sciences and Engineering Research Council of Canada Grant 2016-06172.

\bibliography{bounded}{}
\bibliographystyle{ws-ijfcs}

\end{document}